\documentclass[aps,pra,twocolumn,showpacs,superscriptaddress,nofootinbib]{revtex4-2}   
\usepackage{graphicx}  
\usepackage[caption=false,subrefformat=parens,labelformat=parens]{subfig}
\usepackage{float}
\usepackage[toc,page]{appendix}
\usepackage[table]{xcolor}
\usepackage{mathtools}
\usepackage{dcolumn}   
\usepackage{bm}        
\usepackage{color}

\usepackage{hyperref}
\usepackage{pifont}
\usepackage{mathrsfs}
\usepackage{amsmath, amsthm, amssymb}
\usepackage{amsfonts}
\usepackage{relsize}    
\usepackage{bbold}       
\usepackage{accents}    

\hypersetup{
	colorlinks=true, 
	linktoc=all,     
	linkcolor=blue,  
	citecolor=blue,
	filecolor=blue,
	urlcolor=blue
}
\usepackage{soul}

\usepackage{tikz}
\usetikzlibrary{calc,arrows.meta, decorations.markings,decorations.pathmorphing}

\newcommand{\beq}{\begin{equation}}
	\newcommand{\eeq}{\end{equation}}
\newcommand{\bqa}{\begin{eqnarray}}
	\newcommand{\eqa}{\end{eqnarray}}
\newcommand{\nn}{\nonumber}

\newcommand{\erf}[1]{Eq.~(\ref{#1})}

\newcommand{\erfa}[2]{Eqs.~(\ref{#1}) and (\ref{#2})}
\newcommand{\arf}[1]{{App.}~\ref{#1}} 
\newcommand{\srf}[1]{Sec.~\ref{#1}}

\newcommand{\frf}[1]{Fig.~\ref{#1}}

\newcommand{\dg}{^\dagger}

\definecolor{BLACK}{gray}{0}
\definecolor{RED}{rgb}{1,0,0}
\definecolor{GREEN}{rgb}{0.2,.6,0.2}
\definecolor{amber}{rgb}{1.0,0.29,0.1}

\newcommand{\bra}[1]{\langle{#1}|}

\newcommand{\ket}[1]{|{#1}\rangle}

\newtheorem{theorem}{Theorem}

\newtheorem{proposition}{Proposition}

\newtheorem{lemma}{Lemma}

\begin{document}
	
	\widetext

 \title{From normal Lindbladians to non-normal quantum trajectories}

	\author{Shakib Daryanoosh} 
\email{shakib.daryanoosh@curtin.edu.au}
\affiliation{Curtin Centre for Optimisation and Decision Science, Curtin University, Whadjuk Country, Perth 6102, Australia}

%
%


%
\vskip 0.25cm
	\date{\today}

\begin{abstract}
Efficient simulation of Markovian open quantum systems remains a central challenge because the density-matrix description grows exponentially with system size. Quantum trajectory methods provide an alternative by replacing mixed-state evolution with stochastic pure-state realizations. Here we investigate this framework for normal Lindblad generators, whose orthogonal eigenoperator decomposition precludes transient amplification. By decomposing the Lindbladian into deterministic smooth and stochastic jump contributions, we derive an exact steady-state balance relation that identifies the interplay between these processes as the mechanism underlying Liouvillian normality. We further show that normal Lindbladians exclude exceptional points and that, although individual quantum trajectories generally exhibit stochastic coupling between Liouvillian eigenmodes, these couplings cancel upon ensemble averaging, recovering independent orthogonal relaxation modes. These results provide a trajectory-level interpretation of Liouvillian normality and clarify how a global property of the Lindblad generator is realized through stochastic quantum dynamics.
\end{abstract}

	\maketitle

\section{Introduction} \label{sec:intro}

The dynamics of open quantum systems underpins a broad range of modern quantum
technologies, including quantum sensing, communication, and computation~\cite{BrePet07,BriZol98,LidBru13,DegCap17,VerCir09,Darslu24,CamHap24}. In these settings, irreversible interactions with the
environment are unavoidable and are most commonly described within the Markovian approximation by the Gorini--Kossakowski--Sudarshan--Lindblad (GKSL)
master equation~\cite{GKS76,Lindblad76}. While the master equation provides a
complete statistical description of the reduced system dynamics, many physical
processes are more naturally understood in terms of continuously monitored
quantum trajectories~\cite{DalMol92,DumRit92,GisPer92,Car08,WisMil10,MurSid13}. Beyond their conceptual role, quantum trajectories have become indispensable computational tools, allowing
large open quantum systems to be simulated through stochastic wave-function
methods rather than direct propagation of the full density
matrix~\cite{PenZho25,SanMen25,BorMar26,LiuChe26,SanWil26}.

Recent years have witnessed growing interest in the dynamical properties of
non-Hermitian and dissipative phenomena~\cite{AshUed20, OchFul24, BelDon25,LiMa26}, particularly the role of non-normal operators whose eigenvectors are non-orthogonal despite possessing stable spectra. Non-normality is now recognized as the mechanism underlying transient amplification, pseudospectral sensitivity, and phenomena such as quantum chaos and quantum Mpemba effect~\cite{TreEmb05,Dar26,NavLar26, Longhi25}. These ideas have recently been extended to Markovian quantum dynamics, where the Lindblad generator (or Liouvillian) itself can exhibit non-normal behavior, giving rise to transient amplification and geometric effects that are invisible from the Liouvillian spectrum alone~\cite{Dar26}. This viewpoint introduces a complementary classification of dissipative quantum dynamics based not only on spectral properties but also on the geometry of Liouville space.

Despite these developments, existing analyses have remained almost entirely at
the level of the unconditional master equation. Since experimentally observed open-system dynamics and numerical trajectory simulations are naturally described through stochastic master equations, an important open question is whether and how Lindbladian non-normality manifests itself at the level of individual trajectory realizations. In particular, the deterministic non-Hermitian evolution and stochastic quantum jumps represent distinct
physical processes whose separate contributions to non-normal behavior have not
been systematically investigated. Understanding this decomposition is
important both for interpreting continuously monitored dynamics and for
assessing the complexity of trajectory-based simulation methods.

In this work we develop a trajectory-level theory of Liouvillian normality. Starting from the decomposition of the Lindblad generator into deterministic smooth evolution and stochastic quantum jumps, we show that Liouvillian normality is an emergent property of their algebraic interplay rather than of either contribution individually. Our analysis shows that normality generally does not arise because the smooth dynamics or the jump process is individually normal. Instead, it emerges from their combined action. We further demonstrate that this mechanism acquires a particularly transparent interpretation in the steady-state subspace, where the smooth and jump contributions obey an exact balance relation.

We subsequently investigate the consequences of this structure for quantum
trajectory dynamics by expanding both conditioned and unconditional evolution
in the Liouvillian eigenoperator basis. It is first shown that normal
Lindbladians are necessarily free of Exceptional Points (EPs)~\cite{BenBoe98, Heiss12}. Moreover, while the
unconditional evolution generated by a normal Lindbladian decomposes into
independent Liouvillian modes, stochastic quantum trajectories generally
continue to couple orthogonal modes within individual realizations. This
establishes a clear distinction between ensemble and trajectory dynamics,
showing that modal independence is recovered only after ensemble averaging.

Building on this framework, we further analyze second-moment dynamics in a doubled
Liouville-space representation and show that global normality excludes
transient amplification arising from the uncoupled doubled evolution. Within the doubled-space description, any remaining growth mechanism is associated with the stochastic correlation contribution induced by the unraveling. Finally, we examine two important subclasses---Hermitian Lindbladians and structured dissipative processes---for which the general theory simplifies considerably.

The remainder of the paper is organized as follows. In
Sec.~\ref{sec:framework} we briefly review Liouvillian normality and the
necessary background. Section~\ref{sec:normal_liouvillians} develops the
trajectory-level theory of normality. The geometric consequences for wave-function
evolution and state overlap are investigated in
Sec.~\ref{subsec:wavefunction_geometry}, while
Sec.~\ref{sec:normal:L:features} discusses several general structural features
of normal Lindbladians. In
Sec.~\ref{subsec:single_trajectory_dynamics} we analyze conditional trajectory
dynamics in the Liouvillian eigenoperator basis, and
Sec.~\ref{subsec:variance_bounds} develops the doubled Liouville-space
formalism for second-moment dynamics and trajectory fluctuations. Finally,
Sec.~\ref{sec:specific:scenarios} examines two analytically tractable
subclasses of normal Lindbladians before we conclude in
Sec.~\ref{sec:conclusion}.

\section{Markov open quantum systems and non-normality} \label{sec:framework}
For Markovian open quantum system the master equation is described by the GKSL master equation~\cite{Lindblad76,GKS76}
\begin{equation} \label{eqn:GKSL}
	\dot{\rho} = -i[\hat H,\rho] + \sum_{k=1}^K \left(\hat L_k \rho \hat L_k^\dagger - \frac{1}{2} \{\hat L_k^\dagger \hat L_k, \rho\} \right) \equiv \mathcal L(\rho),
\end{equation}
where $\mathcal L$ is the Lindbladian generator acting on operator space, and $\{\hat L_k\}_{k=1}^K$ are Lindblad operators (we assume $\hbar=1$ throughout). We can assume that any generator $\mathcal L$ admits the decomposition
\begin{equation} \label{L:d:nd}
	\mathcal L = \mathcal L_{\rm d} + \mathcal L_{\rm nd},
\end{equation}
where
\begin{equation} \label{Ld:Lnd}
	\mathcal L_{\rm d} := \frac{1}{2}(\mathcal L + \mathcal L^\dagger), \qquad 
	\mathcal L_{\rm nd} := \frac{1}{2}(\mathcal L - \mathcal L^\dagger).
\end{equation}
The Hermitian $\mathcal L_{\rm d}$ and anti-Hermitian $\mathcal L_{\rm nd}$ components are defined with respect to the Hilbert-Schmidt inner product (which for two operators $\hat A$ and $\hat B$ in finite-dimensional Hilbert space, it is defined as $\big\langle{\hat A,\hat B}\big\rangle \coloneq {\rm Tr}\big[\hat A^\dagger \hat B\big]$), and are denoted as the dissipative and nondissipative parts, respectively..

It was shown in Ref.~\cite{Dar26} that Markovian dynamics generated by a
Lindbladian $\mathcal L$ can be classified according to its degree of
non-normality quantified through 
\begin{equation} \label{eqn:nonnorm:def}
	\eta(\mathcal{L}) \coloneq \left\|\left[\mathcal{L},\mathcal{L}\dg\right]\right\|,
\end{equation}
the dissipation strength
\begin{equation}
	\delta(\mathcal{L}) \coloneq \left\|\mathcal{L}_{\rm d}\right\|
\end{equation}
and the ratio
\begin{equation}
	\kappa(\mathcal{L}) = \frac{\eta(\mathcal{L})}{[\delta(\mathcal{L})]^2}.
\end{equation}
Here and throughout $\|\cdot\|$ denotes the operator norm induced by the Hilbert-Schmidt inner product: 
\begin{equation} \label{eqn:OptNorm:HS}
	\|\mathcal L\| := \sup_{X \neq 0} \frac{\|\mathcal L(X)\|_{\rm HS}}{\|X\|_{\rm HS}}.
\end{equation}

This classification was formulated entirely at the level of the full Lindblad generator. The objective of the present work is to investigate whether the same notions leave observable signatures at the level of quantum trajectories and, if so, whether they influence the complexity of trajectory-based simulations.

\section{Quantum trajectories and non-normality} \label{sec:normal_liouvillians}
While the quantity $\eta(\mathcal L)$ characterizes the full generator and is
independent of the choice of unraveling, a quantum trajectory resolves the
same dynamics into distinct deterministic and stochastic processes. This
raises a natural question: how is global Liouvillian normality reflected in
the individual components of a trajectory? In particular, the smooth
non-Hermitian evolution and stochastic jump events are each capable of
generating non-orthogonal mode dynamics, even when their combined contribution
to the unconditional evolution is normal. To address this question, we
decompose the Lindbladian into its no-jump and jump components and analyze the
corresponding non-normality structure.

The master equation, \erf{eqn:GKSL} can be rearranged such that the Lindbladian is partitioned into a deterministic, non-Hermitian smooth superoperator $\mathcal{S}$ and a stochastic jump superoperator $\mathcal{J}$:
\begin{equation} \label{eqn:so:jump:no-jump}
	\mathcal L \equiv \mathcal{S} + \mathcal{J},
\end{equation}
where the no-jump evolution generator is 
\begin{equation}
	\mathcal{S}(\rho) \equiv -i(\hat H_{\rm eff} \rho - \rho \hat H_{\rm eff}\dg )
\end{equation}
with the non-Hermitian Hamiltonian defined as
\begin{equation} \label{eqn:nonH:Ham}
	\hat H_{\rm eff} = \hat H - \frac{i}{2} \sum_k \hat L_k\dg \hat L_k,
\end{equation}
and the quantum jumps are governed through
\begin{equation}
	\mathcal{J}(\rho) \equiv \sum_k \hat L_k \rho \hat L_k\dg = \sum_k \mathcal{J}_k(\rho).
\end{equation}
An important observation is that the decomposition in \erf{eqn:so:jump:no-jump} depends on the chosen unraveling, that is to say, it is invariant under transformation~\cite{DarBra16}
\begin{subequations} \label{eqn:invar:ME:trans}
	\begin{align}
		\hat{L}_k &\longrightarrow \hat{L}_{m} = \sum_{k=1}^K \widetilde{U}_{mk} \hat{L}_k + \alpha_{m} \hat I, \\
		\hat{H} &\longrightarrow \hat{H} - \frac{i}{2} \sum_{m=1}^{M} \left(\alpha_m^\ast \hat{L}_{m} - \alpha_{m} \hat{L}_{m}\dg\right),
	\end{align}
\end{subequations}
where $\alpha \in \mathbb{C}$, the identity operator is denoted by $\hat I$, and $\widetilde{U}_{mk} \in \mathbb{C}^{M\times K}$ is an arbitrary semi-unitary matrix. In contrast, the full Lindbladian $\mathcal L$ is unraveling independent. 

The corresponding adjoint superoperators with respect to the Hilbert-Schmidt inner product are uniquely determined by the relation $\langle \hat A,{\cal L}(\hat B)\rangle = \langle{\cal L}\dg(\hat A),\hat B\rangle$, yielding:
\begin{align}
	\mathcal{S}^\dagger(\rho) &= i[\hat H, \rho] - \frac{1}{2}\sum_k \{\hat L_k^\dagger \hat L_k, \rho\}, \\
	\mathcal{J}^\dagger(\rho) &= \sum_k \hat L_k^\dagger \rho \hat L_k,
\end{align}
where $\{\hat A, \hat B\}$ denotes anti-commutator for operators $\hat A$ and $\hat B$. The decomposition into $\mathcal S$ and $\mathcal J$ is particularly useful
because the two components have fundamentally different dynamical roles. The
smooth component governs the deterministic evolution between quantum jumps,
while the jump component introduces stochastic state transformations.
Although these contributions combine to produce on average the same unconditional
Lindblad evolution, their potential individual non-normal properties need not coincide
with those of the full generator.

A normal Lindbladian superoperator $\eta(\mathcal{L})=0$ satisfies $\left[\mathcal{L}, \mathcal{L}^\dagger\right] = 0$ which can be expanded out to obtain
\begin{equation} \label{eq:operator_identity}
	\mathcal{D}_{\mathcal{S}} + \mathcal{D}_{\mathcal{J}} + \mathcal{D}_{\mathcal{SJ}} = 0,
\end{equation}
where\footnote{The quantities $\mathcal D_{\mathcal S},\mathcal D_{\mathcal J}$ and
	$\mathcal D_{\mathcal SJ}$ are operator-valued contributions to the Liouvillian commutator. Although $\mathcal D_{\mathcal S}\neq0$ (respectively $\mathcal D_{\mathcal J}\neq0$) implies that $\mathcal S$ (respectively $\mathcal J$) is non-normal, their operator norms do not combine to give the Liouvillian normality measure $\eta(\mathcal L)=\|[\mathcal L,\mathcal L^\dagger]\|$. Only the total commutator determines the global normality of the Lindbladian which satisfies $\|\mathcal{D}_{\mathcal{S}} + \mathcal{D}_{\mathcal{J}} + \mathcal{D}_{\mathcal{SJ}}\| \le \|\mathcal{D}_{\mathcal{S}}\|+\|\mathcal{D}_{\mathcal{J}}\|+\|\mathcal{D}_{\mathcal{SJ}}\|$.}
\begin{subequations} \label{eqn:op:id:separate}
	\begin{align} 
		\mathcal{D}_{\mathcal{S}} &\equiv [\mathcal{S}, \mathcal{S}^\dagger], \\
		\mathcal{D}_{\mathcal{J}} &\equiv [\mathcal{J}, \mathcal{J}^\dagger], \\
		\mathcal{D}_{\mathcal{SJ}} &\equiv [\mathcal{S}, \mathcal{J}^\dagger] + [\mathcal{J}, \mathcal{S}^\dagger].
	\end{align}
\end{subequations}
The expression in \erf{eq:operator_identity} reveals that Lindbladian normality is not generally achieved through simple component-wise cancellation. Instead, the global constraint forces a strict algebraic relation between the individual metrics and a coherent coupling superoperator $\mathcal{D}_{\mathcal{SJ}}$.

Note that due to the transformation presented in \erf{eqn:invar:ME:trans}, the operators $\mathcal{D}_{\mathcal{S}}$, $\mathcal{D}_{\mathcal{J}}$, and $\mathcal{D}_{\mathcal{SJ}}$ are generally unraveling dependent, while the total commutator $\left[\mathcal{L},\mathcal{L}\dg\right]$ and therefore $\eta(\mathcal L)$ remains invariant.

In this framework, $\mathcal{D}_{\mathcal{S}}$ and $\mathcal{D}_{\mathcal{J}}$ represent the independent contributions of the smooth and stochastic components to non-orthogonal mode mixing induced by stochastic map. Crucially, as shown below, these contributions are not canceled directly by one another. Instead, the coupling term $\mathcal{D}_{\mathcal{SJ}}$ provides the compensating contribution required to satisfy the Lindbladian normality condition.

\subsection{Steady-state balance of smooth and jump non-normality} \label{subsec:ss:balance}

The decomposition of the Liouvillian into smooth and stochastic contributions,
raises the question of how the non-normality of the individual components is
related under Liouvillian normality. In particular, while the total
commutator vanishes, $[\mathcal{L},\mathcal{L}^{\dagger}]=0$,
the individual contributions $\mathcal{D}_{\mathcal{S}}$ and
$\mathcal{D}_{\mathcal{J}}$ need not vanish separately. In other words, the operator identity of Eq.~\eqref{eq:operator_identity} constrains the decomposition of Liouvillian normality into its smooth, jump, and mixed contributions. To expose the physical content of this condition, we evaluate it on the asymptotic steady state $\rho_{\rm ss}$ satisfying $\mathcal{L}(\rho_{\rm ss})=0$. The following proposition establishes the steady-state balance relation between these two contributions.

\begin{proposition}[Steady-state balance of smooth and jump non-normality]
	\label{prop:steady_state_balance}
	Let $\mathcal{L}=\mathcal{S}+\mathcal{J}$ be a normal Lindblad generator,
	$[\mathcal{L},\mathcal{L}^{\dagger}]=0$, possessing a steady state
	$\rho_{\rm ss}$ satisfying $\mathcal{L}(\rho_{\rm ss})=0$. Then the smooth and jump commutator contributions satisfy
	\begin{equation} \label{eq:self-commutator:proof}
		\left\langle \rho_{\rm ss},[\mathcal{D}_\mathcal{S}-\mathcal{D}_\mathcal{J}](\rho_{\rm ss})\right\rangle
		= 0.
	\end{equation}
	Equivalently, the cross contribution obeys
	\begin{equation} \label{eq:cross-commutator:proof}
		\left\langle \rho_{\rm ss},[\mathcal{D}_\mathcal{SJ}+2\mathcal{D}_\mathcal{O}](\rho_{\rm ss})\right\rangle = 0, \quad for\;\; \mathcal{O} \in \{\mathcal{S}, \mathcal{J}\}.
	\end{equation}

\end{proposition}
\begin{proof}
The proof follows from evaluating the normality condition in the steady-state
subspace using the Hilbert-Schmidt inner product. The key step is the kernel equivalence of a normal superoperator
\begin{equation}
	\label{eq:kernel_equivalence}
	\mathrm{ker}(\mathcal{L})
	=
	\mathrm{ker}(\mathcal{L}^{\dagger}),
\end{equation}
which implies if $\mathcal L(\rho_{\rm ss})=0$, then $\mathcal L^\dagger(\rho_{\rm ss})=0$, and hence
\begin{subequations} \label{eqn:ss:smooth:jump}
	\begin{align}
	\mathcal{S}(\rho_{\rm ss})
	&=
	-\mathcal{J}(\rho_{\rm ss}),\\
	\mathcal{S}^{\dagger}(\rho_{\rm ss})
	&=
	-\mathcal{J}^{\dagger}(\rho_{\rm ss}).
	\end{align}
\end{subequations}
To convert the operator identity into a scalar constraint on the steady state, we evaluate its Hilbert-Schmidt expectation value
\begin{equation}
	\left\langle \rho_{\rm ss},[\mathcal{L},\mathcal{L}^\dagger](\rho_{\rm ss})\right\rangle=0.
\end{equation}
Substituting the expansion in \erf{eq:operator_identity} and the steady-state kernel relations \erf{eqn:ss:smooth:jump} yields~\erf{eq:cross-commutator:proof} from which the balance expression~\erf{eq:self-commutator:proof} follows immediately. 
\end{proof}

\paragraph*{Physical interpretation.}

Proposition~\ref{prop:steady_state_balance} clarifies how 
Liouvillian normality is realized at the level of individual unraveling
components. The smooth $\mathcal{S}$ and jump terms
$\mathcal{J}$ are not required to be normal independently. Each may possess
a non-vanishing self-commutator and therefore retain the capacity to distort
the dynamical trajectory. The steady-state balance relation,
however, shows that these contributions are not independent. Their
expectation values are locked together by the global normality constraint,
while the mixed smooth--jump contribution supplies the compensating term
required by normality condition.

Consequently, Lindbladian normality does not arise because the no-jump and jump
components are individually free of non-normal effects. Rather, it emerges from a
precise balance between the self-commutator contributions and the
smooth--jump interference term. This cancellation mechanism provides the
structural foundation for the trajectory-stability analysis developed in the
following sections. 

\section{Normal Liouvillians: spectral rigidity and trajectory stability} \label{sec:normal:L:features}
\subsection{Absence of exceptional points} \label{subsec:exceptional_points}

A defining feature of non-normal open quantum systems is their vulnerability to exceptional points. An exceptional point occurs when a parameter-dependent superoperator undergoes a spectral degeneracy where two or more eigenvalues coalesce, and their corresponding eigenoperators simultaneously transform into a single, shared state vector~\cite{MulRot08, MinNor19, AshUed20}. This renders the Liovillian matrix defective (not diagonalizable) and introduces nontrivial Jordan blocks into the dynamics. In quantum trajectory simulations, approaching an EP can lead to critical
algorithmic slowing down, numerical ill-conditioning, and noticeable sampling variance spikes. 

We now prove algebraically that a normal Lindblad superoperator is strictly forbidden from possessing EPs, guaranteeing spectral stability across all parameter spaces.

\begin{theorem} \label{theorem:noEP}
	Let $\mathcal{L}$ be a normal superoperator acting on the Hilbert-Schmidt space, satisfying $\|[\mathcal{L}, \mathcal{L}^\dagger]\| = 0$. The superoperator $\mathcal{L}$ cannot possess an exceptional point and is unitarily diagonalizable.
\end{theorem}

\begin{proof}
	A normal operator on a finite-dimensional Hilbert space satisfies the spectral
	theorem and is therefore unitarily diagonalizable. Consequently, there exists
	an orthonormal basis of eigenoperators
	$\{\rho_\alpha\}$ such that
	\begin{subequations} \label{eqn:normal:L:basis}
		\begin{align}
			\mathcal L(\rho_\alpha) &= \Lambda_\alpha \rho_\alpha, \\
			\mathrm{Tr}\left(\rho_\alpha^\dagger \rho_\beta\right)&=\delta_{\alpha\beta},
		\end{align}
	\end{subequations}
	where $\Lambda_\alpha \in \mathbb{C}$. Since $\mathcal L$ is diagonalizable, it contains no nontrivial Jordan blocks. Exceptional points require precisely such a Jordan structure, in which the algebraic and geometric multiplicities of an eigenvalue differ. Therefore, a normal Lindbladian cannot possess an EP.
\end{proof}
For an alternative proof we refer the interested reader to consult \arf{appn:proof:noEP}.

The absence of exceptional points constitutes a form of spectral rigidity.
Although the smooth and jump components may each be individually non-normal,
global normality forces the full Liouvillian to remain diagonalizable
throughout parameter space. Consequently, trajectory fluctuations cannot be
amplified through Jordan-block dynamics, eliminating one of the principal
mechanisms responsible for critical slowing down and extreme sensitivity in
non-normal open quantum systems.

\subsection{Orthogonal Liouville-space dynamics} \label{sec:orthogonal:dyn}

The absence of exceptional points has an immediate dynamical consequence.
Because a normal Liouvillian satisfies$[\mathcal L,\mathcal L^\dagger]=0$, 
the spectral theorem guarantees the existence of a complete orthonormal eigenoperator basis $\{\rho_\alpha\}$ in Hilbert-Schmidt space satisfying \erf{eqn:normal:L:basis}.
Any density operator may therefore be expanded as
\begin{equation} \label{eqn:rho:expand:rho_a}
	\rho(t) = \sum_\alpha c_\alpha(t)\rho_\alpha .
\end{equation}
The master equation $\dot{\rho}=\mathcal L(\rho)$ then reduces to
\begin{equation} \label{eqn:ME:coef}
	\dot c_\alpha(t)=\Lambda_\alpha c_\alpha(t),
\end{equation}
with solution
\begin{equation} \label{eqn:ME:coef:sol}
	c_\alpha(t) = c_\alpha(0)e^{\Lambda_\alpha t}.
\end{equation}
Thus each Liouvillian mode evolves independently. Unlike non-normal
systems, where non-orthogonal eigenoperators can generate substantial
transient amplification through modal interference, a normal
Liouvillian admits no such mechanism. The dynamics decomposes into a
collection of independent exponential modes whose amplitudes evolve
without cross-coupling.

This orthogonal modal structure provides a spectral interpretation of
the trajectory stability discussed in the following sections. Although
the individual smooth and jump contributions may each be non-normal,
their combined action produces a Liouvillian whose global dynamics is
governed by an orthogonal eigenoperator basis and is therefore free
from the transient amplification associated with defective or highly
non-orthogonal spectra.

\subsection{Modal Compressibility of Normal Liouvillian Dynamics}

The evolution of each modal amplitude is completely independent of all others. In particular, orthogonality of the eigenoperators prevents cross-mode amplification, while the absence of Jordan blocks eliminates algebraic growth factors that would otherwise accompany defective spectra. Consequently, the long-time dynamics are determined solely by the subset of modes possessing the smallest decay rates. Rapidly decaying modes become exponentially suppressed and may be truncated without introducing hidden pseudospectral amplification effects.

To quantify this compression, define the $\epsilon$-effective mode rank as
\begin{equation}
	R_\epsilon(t)= \min \bigg\{m:\sum_{\alpha=1}^{m}|c_\alpha(t)|^2 \ge (1-\epsilon) \sum_{\beta}|c_\beta(t)|^2 \bigg\}.
\end{equation}
The quantity $R_\epsilon(t)$ measures the number of Liouvillian modes required to represent the unconditional state to accuracy $\epsilon$. Since each coefficient decays independently according to its physical relaxation rate, the effective dimension of the dynamics is determined entirely by the density of slowly decaying modes rather than by transient mode mixing.

This result should not be confused with the behavior of individual quantum trajectories. We will see in \srf{subsec:single_trajectory_dynamics} for a given unraveling, the trajectory coefficients $a_\alpha(t)$ generally remain coupled through the stochastic jump process, \erf{eqn:sme:coeff}, even when the underlying Liouvillian is normal. The simplification occurs only after ensemble averaging, where the jump-induced couplings cancel and the unconditional coefficients satisfy the decoupled evolution equation above. Normality therefore guarantees modal compressibility at the ensemble level, while the trajectory-level manifestations of non-normality are constrained through the steady-state balance relations established in the preceding sections.

\section{Wave-Function Dynamics: roles of Smooth Evolution and Jump Events}
\label{subsec:wavefunction_geometry}

The steady-state balance expression presented in \srf{subsec:ss:balance} establishes that Liouvillian normality generally emerges from a precise cancellation between the commutation contributions of the smooth, jump, and mixed contributions. This result, however, is formulated entirely in Liouville space and does not immediately reveal how the cancellation manifests at the level of individual quantum trajectories. To develop a more intuitive picture, we now examine the evolution of wave functions in the underlying system's  Hilbert space.

Our goal is to understand how the two constituents of the unraveling act on trajectory space. We therefore consider two initially orthogonal states, $\vert\phi_1(0)\rangle$ and $\vert\phi_2(0)\rangle$, satisfying $\langle\phi_1(0)\vert\phi_2(0)\rangle=0$, and study how their relative orientation evolves under the no-jump dynamics and under stochastic jump events. As we shall see, the smooth evolution continuously reshapes trajectory orientations in Hilbert space, while jump events induce discrete, state-dependent transformations of the same geometric structure, and it is the interplay between the two that prevents trajectory-level non-orthogonality from surviving in the ensemble-averaged dynamics.

The interpretation developed below is most relevant when the individual self-commutators possess nonvanishing steady-state expectation values. In this situation, the mixed contribution is required to cancel the non-normality generated by the individual unraveling components. In the special case where these expectations vanish, the balance relation is satisfied trivially and no compensating distortion between the smooth and jump components is required.

\subsubsection{Smooth evolution analysis}

To understand how the steady-state balance identity manifests at the level of individual trajectories, we first consider the deterministic evolution between stochastic jump events. During a no-jump interval, an unnormalized wave function evolves according to the effective non-Hermitian Hamiltonian $\hat H_{\rm eff}$, introduced in Eq.~\eqref{eqn:nonH:Ham}. The corresponding state-vector dynamics is
\begin{equation}
	|\phi_m(t)\rangle = e^{-i \hat H_{\rm eff}t}|\phi_m(0)\rangle .
\end{equation}
To quantify the geometric action of the smooth evolution, we examine the overlap
\begin{equation}
	O(t) \equiv \langle \phi_1(t)|\phi_2(t)\rangle .
\end{equation}
Differentiating with respect to time gives
\begin{equation}
	\dot{O}(t) = \langle\phi_1(t)|(i[\hat H_{\rm eff}^{\dagger} - \hat H_{\rm eff}]) |\phi_2(t)\rangle.
\end{equation}
The first derivative depends only on the anti-Hermitian part of the effective Hamiltonian and therefore isolates the dissipative contribution of the no-jump evolution. For the effective Hamiltonian~(\ref{eqn:nonH:Ham}), this term is generated by the positive operator $\sum_k \hat L_k^\dagger \hat L_k$ and governs the attenuation of unnormalized trajectory amplitudes, while the Hermitian part contributes only norm-preserving coherent evolution. Taking a second derivative separates the purely dissipative contraction from the genuinely non-normal contribution. A straightforward calculation yields
\begin{equation}
	\ddot{O}(t) = -\langle\phi_1(t)| \Big[ (\hat H_{\rm eff}^{\dagger}-\hat H_{\rm eff})^2 + [\hat H_{\rm eff},\hat H_{\rm eff}^{\dagger}] \Big] |\phi_2(t)\rangle .
\end{equation}
The first term is present for any non-Hermitian evolution and reflects the irreversible loss associated with the no-jump dynamics. The second term depends explicitly on the non-normality of the effective Hamiltonian and vanishes only when $\hat H_{\rm eff}$ is normal. Consequently, the curvature of the overlap dynamics is directly sensitive to the non-normal geometry of the no-jump evolution.

The origin of this effect becomes particularly transparent in the eigenbasis of the effective Hamiltonian. For a non-normal operator, the right eigenvectors
$\{|r_m\rangle\}$ are generally non-orthogonal and satisfy the biorthogonality relation
\begin{equation}
	\langle \ell_m | r_n\rangle = \delta_{mn},
\end{equation}
with the left eigenvectors $\{|\ell_m\rangle\}$.
The spectral decomposition therefore takes the form
\begin{equation}
	\hat H_{\rm eff} = \sum_m \lambda_m |r_m\rangle \langle \ell_m|,
\end{equation}
where
\begin{equation}
	\lambda_m = \varepsilon_m - \frac{i}{2}\gamma_m,
\end{equation}
are complex eigenvalues.
Expanding two arbitrary initial states as
\begin{equation}
	|\phi_j(0)\rangle=\sum_m c_{jm} |r_m\rangle, \quad {\rm for}\;\, j \in \{1,2\}, 
\end{equation}
their no-jump evolution becomes
\begin{subequations}
\begin{equation}
	|\phi_j(t)\rangle = \sum_m c_{jm} \,e^{-i \varepsilon_m t}\, e^{-\gamma_m t/2}\, |r_m\rangle.
\end{equation}
\end{subequations}
Let $\gamma_0$ denote the smallest decay rate. On time scales satisfying $(\gamma_m -\gamma_0)t \gg 1$ for all $m>0$, but remaining within a given no-jump interval, all faster-decaying components become exponentially suppressed and the evolution is dominated by the corresponding right eigenvector $|r_0\rangle$. After normalization\footnote{Since the no-jump evolution is naturally described by unnormalized state vectors, the norm decreases according to the survival probability. To isolate the evolution of the state direction, we therefore consider the normalized overlap.},
\begin{equation}
	\frac{|\langle\phi_1(t)|\phi_2(t)\rangle|}{\|\phi_1(t)\| \,\|\phi_2(t)\|} \approx 1.
\end{equation}
Thus, irrespective of their initial orientation, generic no-jump trajectories become progressively aligned with the slowest-decaying mode. This alignment results from the spectral filtering induced by the non-Hermitian effective Hamiltonian: rapidly decaying components are exponentially suppressed, leaving the longest-lived mode dominant over sufficiently long no-jump intervals.
When $\hat H_{\rm eff}$ is non-normal, the non-orthogonality of its eigenvectors \emph{further} modifies the approach to this asymptotic alignment\footnote{Note that even if $\hat H_{\rm eff}$ is perfectly normal, spectral filtering still aligns every no-jump state with the slowest-decaying eigenvector. What is unique to non-normality is \emph{how} the trajectory approaches that asymptotic direction, because the eigenvectors are no longer orthogonal.}. The commutator
$[\hat H_{\rm eff},\hat H_{\rm eff}^\dagger]$ therefore quantifies the genuinely non-normal contribution to the overlap dynamics and provides the trajectory-level counterpart of the smooth-sector commutator appearing in the steady-state balance relation.

\subsubsection{Stochastic Jump Action and smooth--Jump Interference}
The modification of trajectory-state overlaps is not unique to the smooth evolution. Quantum jumps also alter the relative orientation of states within the trajectory ensemble. Consider a jump event occurring through channel $k$ at time $t$. The trajectory states are transformed according to

\begin{equation}
	|\phi_1^{\rm jump}\rangle = \hat L_k|\phi_1(t)\rangle, \qquad |\phi_2^{\rm jump}\rangle = \hat L_k|\phi_2(t)\rangle .
\end{equation}
The overlap immediately after the jump is therefore
\begin{equation}
	\label{eq:jump_overlap} \langle \phi_1^{\rm jump}|\phi_2^{\rm jump}\rangle = \langle\phi_1(t)| \hat L_k^\dagger \hat L_k |\phi_2(t)\rangle .
\end{equation}
Equation~\eqref{eq:jump_overlap} shows that the jump process acts through the positive operator $\hat L_k^\dagger \hat L_k$, which generally modifies the relative orientation of trajectory states in Hilbert space. Unless $\hat L_k^\dagger \hat L_k$ is proportional to the identity, the jump update is therefore nonunitary and can change the relative overlap between states within the trajectory ensemble.

This behavior parallels the smooth evolution discussed above. The deterministic drift is generated by the non-Hermitian propagator $\exp({-i \hat H_{\rm eff}t})$, while the jump dynamics is generated by the operators $\hat L_k$. Both are nonunitary transformations and, at the superoperator level, their associated self-commutators quantify the corresponding smooth and jump contributions to non-orthogonal mode mixing induced by stochastic map.

The significance of the steady-state balance relation derived in \srf{subsec:ss:balance} is that these two sources of non-normality do not act independently. Taken together, the smooth evolution continuously reshapes trajectory overlaps between jumps, while the jump events induce discrete overlap transformations. The steady-state balance proposition shows that these two mechanisms are not independent: although each may generate non-normal trajectory dynamics on its own, their combined effect is constrained by the mixed smooth--jump interference term so that the full Liouvillian remains normal.

\section{Single-Trajectory Dynamics in the Liouvillian Eigenbasis} \label{subsec:single_trajectory_dynamics}
The orthogonality of the Liouvillian eigenoperators has an immediate consequence for the ensemble-averaged dynamics: distinct Liouvillian eigenmodes evolve independently at the level of the unconditional state.
However, the situation is different at the level of an individual quantum trajectory. Let
\begin{equation}
	\rho_c(t) = \sum_\alpha a_\alpha(t)\rho_\alpha
\end{equation}
denote the conditioned (normalized) state associated with a single realization of the jump process. By construction,
\begin{equation} \label{eqn:ME:coef:avg}
	c_\alpha(t)=\mathbb{E}[a_\alpha(t)],
\end{equation}
where the expectation value is taken over the trajectory ensemble. For the normalized jump unraveling,
\begin{equation} \label{eqn:sme:coeff}
	d\rho_c =\left(\mathcal{S}(\rho_c)+\sum_k \wp_k\,\rho_c\right)dt+
	\sum_k dN_k \left( \frac{\mathcal{J}_k(\rho_c)} {\wp_k}-\rho_c \right),
\end{equation}
projection onto the Liouvillian eigenbasis yields\footnote{Utilizing the definition of the superoperator adjoint, alongside the normal operator eigenvalue relation $\mathcal{L}^\dagger(\rho_\alpha) = \Lambda_\alpha^* \rho_\alpha$, the linear mapping $\text{Tr}\big(\rho_\alpha^\dagger \mathcal{L}(\rho_{c})\big)$ simplifies exactly to $ \Lambda_\alpha a_\alpha(t)$.}
\begin{equation}
	da_\alpha
	=\Big[\wp_\alpha a_\alpha- \sum_k \mho_{\alpha k}\Big]dt + \sum_k dN_k \left[\frac{\mho_{\alpha k}}{\wp_k}-a_\alpha \right],
\end{equation}
where
\begin{align}
	\wp_\alpha (t) &= \Lambda_\alpha + \sum_k \wp_k \equiv \Lambda_\alpha + \sum_k\mathrm{Tr}\big(\mathcal{J}_k(\rho_c)\big), \\
	\mho_{\alpha k} &= \mathrm{Tr}\big(\rho_\alpha^\dagger \mathcal{J}_k(\rho_c)\big).
\end{align}
Unlike the deterministic coefficients $c_\alpha(t)$, the trajectory coefficients $a_\alpha(t)$ are generally coupled. The coupling originates from the nonlinear jump update and depends on the instantaneous conditioned state. Consequently, normality of the Liouvillian does not imply independent modal evolution along individual stochastic realizations.

Instead, normality constrains the dynamics only after ensemble averaging. Although the jump process continuously mixes Liouvillian modes within each trajectory, these stochastic couplings cancel in the ensemble average, so that the unconditional coefficients evolve diagonally, recovering \erf{eqn:ME:coef} and ensuring that the orthogonal Liouvillian modes remain dynamically independent at the level of the master equation.

\section{Monte Carlo Efficiency of Normal Liouvillians}
\label{subsec:variance_bounds}

The absence of exceptional points establishes spectral stability of the single-copy Liouvillian. We now show that global normality also constrains the statistical efficiency of stochastic quantum trajectory simulations. Specifically, we analyze the intrinsic trajectory variance associated with physical observables and show that it remains bounded at long times, preventing the sampling inflation that typically accompanies non-normal open-system dynamics.

Let $\hat A$ be a Hermitian observable and define its single measurement outcome as
\begin{equation} \label{eqn:traj:ens:avg}
	\hat A_j(t) = \mathrm{Tr}[\hat A\,\rho_{c_j}(t)],
\end{equation}
where $\rho_{c_j}(t)$ denotes a normalized conditional state generated by a chosen unraveling of the Lindblad equation and indexed by $j$. A Monte Carlo simulation with $N$ independent trajectories produces realizations $\{\hat A_j(t)\}_{j=1}^{N}$ and estimates the ensemble expectation value through

\begin{equation} \label{eqn:MC:ens:avg}
	\bar A_N(t)=\frac{1}{N}\sum_{i=1}^{N} \hat A_j(t).
\end{equation}
The statistical uncertainty of this estimator is determined by the trajectory variance
\begin{equation} \label{eq:path_variance_def}
	\sigma_{\rm traj}^2(t) \equiv \mathrm{Var}\left[\hat A(t)\right] = \mathbb{E}\left[\hat A(t)^2\right] - \mathbb{E}\left[\hat A(t)\right]^2.
\end{equation}
Throughout this section we distinguish between two ensemble averages. The notation
$\mathbb E[\cdot]$ denotes averaging over the normalized conditional-state trajectories generated by the physical unraveling, whereas $\tilde{\mathbb E}[\cdot]$ denotes averaging over the corresponding unnormalized conditional-state ($\tilde{\rho}_c$) unraveling. Although these averages are taken with respect to different trajectory measures, they reproduce the same unconditional density operator,
\begin{equation}
	\rho(t) = \mathbb E[\rho_c(t)] = \tilde{\mathbb E}[\tilde\rho_c(t)].
\end{equation}
For independent trajectories, the Central Limit Theorem gives
\begin{equation} \label{eqn:var:clt}
	\mathrm{Var}\left[\hat{\bar A}_N(t)\right]=\frac{\sigma_{\rm traj}^2(t)}{N}.
\end{equation}
Therefore, achieving a fixed statistical accuracy $\epsilon$ requires a trajectory number scaling as
\begin{equation} \label{eqn:traj:no:scale}
	N\sim \frac{\sigma_{\rm traj}^2(t)}{\epsilon^2}.
\end{equation}
Consequently, bounding the long-time behavior of $\sigma_{\rm traj}^2(t)$ is equivalent to determining whether the Monte Carlo sampling cost remains constant or grows with simulation time.

\subsection{Second-Moment Dynamics in Doubled Liouville Space}
The trajectory variance defined above is a second-moment quantity. While the first moment is directly reproduced by the Lindblad master equation,
\begin{equation}
	\mathbb{E}\left[\hat A(t)\right]=
	 \mathbb{E}\left[\mathrm{Tr}\left(\hat A\,\rho_c(t)\right)\right]=
	\mathrm{Tr}\left(\hat A\,\rho(t)\right),
\end{equation}
the second moment contains additional information about fluctuations between individual stochastic realizations:
\begin{equation}
	\mathbb{E}\left[\hat A(t)^2\right]=
	\mathbb{E}\left[
	\mathrm{Tr}\left(\hat A\,\rho_c(t)\right) \,
	\mathrm{Tr}\left(\hat A\,\rho_c(t)\right)
	\right].
\end{equation}
To analyze this second moment, we introduce a doubled Liouville-space representation. Using the tensor-product trace identity,

\begin{equation}
	\mathrm{Tr}\left(\hat A\,\rho_c\right) \mathrm{Tr}\left(\hat A\,\rho_c\right) = \mathrm{Tr} \left[\left(\hat A\otimes \hat A\right) (\rho_c\otimes\rho_c)\right],
\end{equation}
the second moment can be written as
\begin{equation}
	\mathbb{E}\left[\hat A(t)^2\right]= \mathrm{Tr} \left[ \left(\hat A\otimes \hat A\right) \bm{\varrho}(t) \right],
\end{equation}
where the weighted doubled state is defined as
\begin{equation} \label{eqn:rho:squared:normal}
	\bm{\varrho}(t) \equiv \mathbb{E} \left[ \rho_c(t)\otimes\rho_c(t) \right].
\end{equation}
Therefore, bounding the trajectory variance is equivalent to controlling the long-time behavior of this doubled object.

The difficulty is that normalized quantum trajectories do not obey a linear stochastic evolution. The normalization condition, $\mathrm{Tr}[\rho_c(t)]=1,$ introduces state-dependent nonlinear terms into the stochastic equation of motion. Consequently, the evolution equation for $\bm{\varrho}(t)$ does not generally close into a linear Liouville-space generator.

To expose the underlying linear structure, we temporarily introduce the corresponding unnormalized trajectory state $\tilde{\rho}_c(t)$, related to the normalized state through
\begin{equation} \label{eqn:unnorm:rho}
	\tilde{\rho}_c(t) = w_c(t)\rho_c(t),
\end{equation}
where 
\begin{equation} \label{eqn:weight:unnorm}
w_c(t)=\mathrm{Tr}[\tilde{\rho}_c(t)],
\end{equation}
is the stochastic trajectory weight. The unnormalized state evolves linearly under the stochastic unraveling, allowing the doubled operator
\begin{equation} \label{eqn:rho:squared:unnorm}
	\tilde{\bm{\varrho}}(t) \equiv \tilde{\mathbb{E}} \left[ \tilde{\rho}_c(t)\otimes\tilde{\rho}_c(t) \right]
\end{equation}
to satisfy a closed linear evolution equation of the form given in \erf{eqn:EOM:unnorm:R}.

\subsection{Stability of the Doubled Liouvillian} \label{subsec:doubled_liouvillian_stability}

The unnormalized representation preserves the unconditional density operator through the weighted ensemble average, while allowing the tensor product quantity in \erf{eqn:rho:squared:unnorm} to evolve according to a linear doubled Liouvillian. Using the unnormalized jump stochastic master equation derived in Appendix~\ref{app:doubled_liouvillian}, the doubled state satisfies
\begin{equation} \label{eqn:EOM:unnorm:R}
	\dot{\tilde{\bm{\varrho}}}(t) = \widetilde{\mathcal L}^{(2)}\, \tilde{\bm{\varrho}}(t),
\end{equation}
where
\begin{equation}
	\widetilde{\mathcal L}^{(2)}=\widetilde{\mathcal L}^{(2)}_0+\widetilde{\mathcal W}.
\end{equation}
The uncoupled contribution is
\begin{equation}
	\widetilde{\mathcal L}^{(2)}_0 = \mathcal L\otimes\mathcal I + \mathcal I\otimes\mathcal L ,
\end{equation}
which represents independent evolution of the two copies. The remaining term,
\begin{equation}
	\widetilde{\mathcal W} = \sum_{k=1}^K (\mathcal J_k - \mathcal{I)} \otimes (\mathcal J_k - \mathcal{I)},
\end{equation}
contains all noise-induced correlations between the two tensor copies. The importance of this decomposition is that the first term inherits the spectral properties of the original Liouvillian. In particular, if $	[\mathcal L,\mathcal L^\dagger]=0$, then $\widetilde{\mathcal L}^{(2)}_0$ is also normal 
\begin{equation}
	\left[\widetilde{\mathcal L}^{(2)}_0,\widetilde{\mathcal L}^{(2)\dagger}_0\right]=0.
\end{equation}
Therefore, the spectral theorem guarantees the existence of an orthonormal tensor-product eigenbasis $\{\rho_\alpha\otimes\rho_\beta\}$ satisfying
\begin{equation}
	\widetilde{\mathcal L}^{(2)}_0
	(\rho_\alpha\otimes\rho_\beta) = \Lambda_{\alpha\beta} (\rho_\alpha\otimes\rho_\beta),
\end{equation}
with eigenvalues $\Lambda_{\alpha\beta} = \Lambda_\alpha+\Lambda_\beta$. The corresponding propagator therefore admits the orthogonal spectral
representation
\begin{equation}
	e^{t\widetilde{\mathcal L}^{(2)}_0} = \sum_{\alpha,\beta}
	e^{(\Lambda_\alpha+\Lambda_\beta)t}\,\Pi_{\alpha\beta},
\end{equation}
where $\Pi_{\alpha\beta}$ are mutually orthogonal spectral projectors. For a normal operator, the operator norm of the propagator is determined
entirely by the spectral abscissa,
\begin{equation}
	\left\| e^{t\widetilde{\mathcal L}^{(2)}_0} \right\| = \max_{\alpha,\beta}\; e^{\left[\mathrm{Re} (\Lambda_\alpha+\Lambda_\beta) \right] t}.
\end{equation}
Because every Lindbladian generates a completely positive trace-preserving
contraction semigroup, its spectrum lies in the closed left half-plane,
implying $\mathrm{Re}(\Lambda_{\alpha\beta})\le 0$. Consequently, $	\left\|\exp\left[t\widetilde{\mathcal L}^{(2)}_0\right]\right\|\le 1$.
Hence the uncoupled doubled dynamics is contractive and cannot exhibit
transient amplification. Any possible growth of second moments must therefore
originate from the stochastic coupling term
$\widetilde{\mathcal W}$.

The absence of transient amplification in the uncoupled doubled dynamics does not, by itself, guarantee stability of the full second-moment evolution. The stochastic unraveling introduces correlations between the two copies of the system, which are encoded in the coupling contribution $\widetilde{\mathcal W}$. To analyze this contribution, it is useful to recall why the unnormalized representation was introduced.

For normalized trajectories, the stochastic state satisfies a nonlinear evolution equation due to the instantaneous normalization after each jump event. Consequently, although the physical density matrix is recovered through the ensemble average, the second moment in \erf{eqn:rho:squared:normal}
does not evolve under a closed linear superoperator. This prevents a direct spectral analysis of the second-moment dynamics.

The unnormalized representation, \erf{eqn:rho:squared:unnorm}, avoids this difficulty.
The advantage of this representation is that the stability of the weighted second moment can be studied using standard Liouville-space techniques. However, the unnormalized state carries a stochastic trajectory weight, \erf{eqn:weight:unnorm}, such that the physical normalized trajectory is recovered as given in \erf{eqn:unnorm:rho}. Therefore, the stability of $\tilde{\bm{\varrho}}(t)$ must be interpreted together with the evolution of the trajectory weights.

The connection between the doubled state and the weight fluctuations follows directly from the identity sector of the doubled space. Using the tensor-product trace identity $\mathrm{Tr}(X\otimes Y) = \mathrm{Tr}(X) \,\mathrm{Tr}(Y)$, which holds for arbitrary linear operators $X$ and $Y$, we obtain 
\begin{equation}
	\mathrm{Tr} \left[ \tilde{\bm{\varrho}}(t) \right] = \tilde{\mathbb{E}} \left[ w_c(t)^2 \right].
\end{equation}
Thus, any growth of the doubled-space trace norm generated by $\widetilde{\mathcal W}$ corresponds directly to fluctuations of the stochastic trajectory weights rather than to an instability of the physical density matrix evolution (\arf{appn:wieght:fluc}).
To make this explicit, taking the trace of the doubled equation of motion, \erf{eqn:EOM:unnorm:R}, gives
\begin{equation}
	\frac{d}{dt} \tilde{\mathbb{E}} \left[w_c(t)^2 \right]=\mathrm{Tr} \left[ \widetilde{\mathcal W} \tilde{\bm{\varrho}}(t) \right],
\end{equation}
where the uncoupled contribution vanishes because the original Lindblad generator is trace preserving. Therefore, the entire deviation from independent contractive evolution is contained in the stochastic correlation part.

This separation clarifies the distinct roles of the two contributions to the doubled dynamics: the uncoupled Liouvillian governs the geometric evolution in doubled Liouville space, while the correlation superoperator governs fluctuations of the stochastic trajectory weights.

Tracing the doubled-space evolution shows that the correlation superoperator
$\widetilde{\mathcal W}$ exclusively governs the growth of trajectory-weight
fluctuations. In particular,
\begin{equation}
	\frac{d}{dt}\tilde{\mathbb E}\left[w_c^2\right]=
	\tilde{\mathbb E} \left[ w_c^2 \sum_k (r_k-1)^2 \right],
\end{equation}
where $r_k=\mathrm{Tr}(L_k^\dagger L_k\rho_c)$ are the instantaneous jump
rates of the normalized trajectory. Since $0\le r_k\le|L_k^\dagger L_k|$
for finite-dimensional systems, the growth rate of
$\tilde{\mathbb E}[w_c^2]$ remains bounded by a finite constant determined by
the jump operators. Consequently, trajectory-weight fluctuations can grow at
most exponentially in time and cannot exhibit super-exponential amplification.
The derivation is given in Appendix~\ref{appn:wieght:fluc}.
In general, the exponential bound does not imply uniformly bounded weight fluctuations. In what follows, we therefore additionally assume that the chosen unraveling satisfies
\begin{equation}
	\sup_t \tilde{\mathbb E} \left[w_c(t)^2 \right] < \infty,
\end{equation}
implying the trajectory-weight component cannot generate unbounded amplification of the observable second moment. Combined with the absence of transient
amplification in the uncoupled doubled dynamics, this removes the two
mechanisms identified above that could otherwise inflate trajectory
fluctuations.

The final step is to connect this result to the physical observable variance. Since
\begin{equation}
	\hat A(t)=\mathrm{Tr}[\hat A\rho_c(t)]={\mathrm{Tr}[\hat A\tilde{\rho}_c(t)]}/{w_c(t)},
\end{equation}
the normalized second moment differs from the linear doubled moment only through the trajectory weights. Under these additional regularity assumptions on the trajectory weights, the stability of the doubled evolution is inherited by the normalized observable moments.

\subsection{Physical Implications}
The preceding analysis establishes rigorous properties of the doubled Liouvillian and trajectory-weight dynamics. We now discuss the implications of these results for the variance of observables and the efficiency of Monte Carlo trajectory simulations under the additional regularity assumptions introduced below.

\subsubsection{Connection to Trajectory Variance}
\label{subsec:observable_variance_connection}

The doubled-space construction provides a linear evolution equation for the
weighted second moment of the stochastic trajectories. The remaining step is
to relate this quantity to the variance of observables evaluated on
normalized trajectories.

The normalized trajectory is recovered from the unnormalized state through \erf{eqn:unnorm:rho}. Therefore, for a Hermitian observable $\hat A$, the second moment can be written as
\begin{equation} 	\label{eq:normalized_second_moment_weight}
	\mathbb{E}\left[\hat A(t)^2\right] = \mathbb{E} \left[
	{\mathrm{Tr} \left[ (\hat A\otimes \hat A) (\tilde{\rho}_c\otimes\tilde{\rho}_c) \right]}/{w_c^2(t)} \right],
\end{equation}
where the numerator is the trajectory-level tensor-product observable whose ensemble average is governed by the doubled Liouvillian $\tilde{\mathcal{L}}^{(2)}$, whereas the denominator arises solely from the normalization map from unnormalized to normalized trajectories.
Equation~(\ref{eq:normalized_second_moment_weight})
shows that the physical trajectory variance receives contributions from two
distinct mechanisms. The first is amplification associated with the
doubled propagator. For globally normal Liouvillians this mechanism is absent
in the uncoupled doubled dynamics because the corresponding generator possesses
an orthogonal eigenbasis and cannot exhibit transient growth. 

The second is the stochastic reweighting induced by trajectory weights, whose evolution is governed entirely by the correlation superoperator $\widetilde{\mathcal W}$. Importantly, growth of the trajectory weights alone does not imply growth of
the physical trajectory variance, because the normalization factor appears in
the denominator of the observable moments. Therefore, the relevant quantity is
not the absolute magnitude of $w_c$, but rather the conditioning of the
normalization process and the absence of anomalously small trajectory weights.

Therefore, the relevant requirement for stable trajectory sampling is not
the boundedness of the weight second moment itself, but the absence of
pathological inverse-weight fluctuations. For unravelings with well-conditioned
trajectory weights, the stability of the doubled evolution transfers to the
normalized observable moments.

 Consequently, if the trajectory normalization remains well conditioned at long times, so that inverse-weight fluctuations do not introduce additional amplification, then the observable second moments inherit the stability of the doubled evolution. Under these assumptions, the trajectory variance remains bounded, preventing the growth of statistical fluctuations associated with long-time sampling.

\subsubsection{Implications for Monte Carlo sampling}

The preceding analysis establishes that, for a normal Liouvillian,
the principal source of trajectory-sampling inflation associated with
non-normal transient amplification is absent. The orthogonal spectral structure of the Liouvillian prevents transient amplification in the uncoupled doubled dynamics, while the remaining stochastic contribution enters exclusively through the trajectory-weight fluctuations.

For an observable $\hat A$, the Monte Carlo estimator constructed from
$N$ independent trajectories satisfies, \erf{eqn:var:clt}. Consequently, achieving a prescribed statistical accuracy $\epsilon$ requires a trajectory number scaling as \erf{eqn:traj:no:scale}. The analysis of the doubled Liouvillian shows that, under the conditions established above, the trajectory variance remains bounded at long times.

Under the regularity assumptions above, normality removes one important mechanism by which the sampling cost can increase with simulation time.
Therefore, the computational effort associated with ensemble sampling
remains asymptotically stable, in contrast to non-normal systems where
transient amplification can inflate trajectory fluctuations and increase
sampling requirements.

This result should be interpreted as a statement about the statistical
efficiency of trajectory ensembles rather than the dynamics of individual
trajectories. Although the coefficients of a single stochastic realization
generally remain coupled through the jump process, these trajectory-level couplings are absent from the ensemble-averaged dynamics, recovering the independent modal evolution of the unconditional density operator.

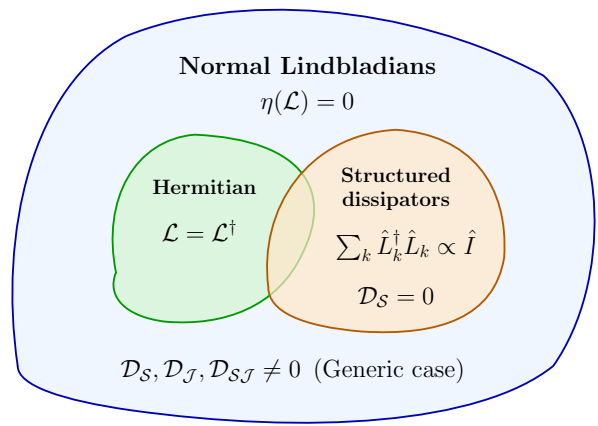
\begin{figure}[t]
	\centering
	\begin{tikzpicture}[scale=0.65,transform shape]
		
		\definecolor{OuterBlue}{RGB}{220,235,255}
		\definecolor{GreenBlob}{RGB}{215,245,215}
		\definecolor{OrangeBlob}{RGB}{255,232,200}
		
		
		\filldraw[
		fill=OuterBlue,
		draw=blue!70!black,
		line width=0.7pt,
		fill opacity=.45
		]
		(-5.9,-2.9)
		.. controls (-6.3,-0.2) and (-5.6,2.8)
		.. (-3.7,3.8)
		.. controls (-1.5,4.7) and (2.0,4.5)
		.. (4.8,3.0)
		.. controls (6.1,1.8) and (6.2,-1.2)
		.. (5.0,-2.9)
		.. controls (3.5,-4.1) and (0.0,-4.3)
		.. (-3.6,-3.8)
		.. controls (-5.0,-3.6) and (-5.9,-3.3)
		.. cycle;
		
		\node[font=\Large\bfseries] at (0,3.2)
		{Normal Lindbladians};
		
		\node[font=\Large\bfseries] at (0,2.45)
		{$\eta(\mathcal L)=0$};
		
		
\filldraw[
fill=GreenBlob,
draw=green!60!black,
line width=0.7pt,
fill opacity=.80
]
(-3.9,-1.0)
.. controls (-4.2,0.5) and (-3.5,1.7)
.. (-2.3,1.8)
.. controls (-0.0,1.7) and (0.5,0.6)
.. (-0.1,-0.6)
.. controls (-0.6,-1.7) and (-1.7,-2.1)
.. (-2.9,-2.0)
.. controls (-3.7,-1.9) and (-4.1,-1.5)
.. cycle;

\node[align=center,font=\large\bfseries]
at (-2.1,0.75)
{Hermitian};

\node[font=\Large\bfseries] at (-2.2,-0.15)
{$\mathcal L=\mathcal L^\dagger$};
		
		
\filldraw[
fill=OrangeBlob,
draw=orange!70!black,
line width=0.7pt,
fill opacity=.70
]
(-0.8,-1.3)
.. controls (-1.0,0.3) and (0.0,1.8)
.. (1.8,1.9)
.. controls (3.3,1.8) and (4.1,0.8)
.. (4.0,-0.7)
.. controls (3.8,-1.9) and (2.7,-2.3)
.. (1.1,-2.2)
.. controls (-0.1,-2.1) and (-0.8,-1.8)
.. cycle;

\node[align=center,font=\large\bfseries]
at (1.8,0.7)
{Structured\\dissipators};

\node[font=\Large\bfseries] at (2.0,-0.45)
{$\sum_k \hat L_k^\dagger \hat L_k\propto \hat I$};
\node[font=\Large\bfseries] at (1.8,-1.5)
{$\mathcal{D}_{\mathcal{S}}=0$};

\node[align=center,font=\Large]
at (-0.3,-3.0)
{$\mathcal{D}_{\mathcal{S}},\mathcal{D}_{\mathcal{J}},\mathcal{D}_{\mathcal{SJ}}\neq 0\,$
	 (Generic case)};
		
	\end{tikzpicture}
	
	\caption{Conceptual organization of the normal Lindbladians considered in this work. The two colored regions represent analytically tractable subclasses whose intersection is nonempty.} \label{fig:normal:L}
\end{figure}

\section{Special scenarios} \label{sec:specific:scenarios}
The general results developed in the preceding sections apply to arbitrary
normal Lindbladians. It is nevertheless instructive to examine concrete
physical examples, both to illustrate the theoretical framework and to
highlight the additional structure present in many experimentally relevant
models. As will become apparent, several commonly encountered dissipative
systems belong to more restrictive subclasses of normal Lindbladians, such as
Hermitian generators or structured dissipative processes. In contrast,
constructing genuinely generic normal Lindbladians, for which the smooth,
jump, and mixed commutator contributions are all individually nonvanishing
while collectively satisfying the normality condition \erf{eq:operator_identity}, 
appears to be a nontrivial problem in its own right, \frf{fig:normal:L}. The examples below therefore illustrate both the general theory and the structural constraints commonly encountered in physically motivated models.

\subsection{Hermitian Lindbladian}

A particularly instructive special case arises when the Lindblad generator is Hermitian in Liouville space,
\begin{equation} \label{eqn:Hermitian:L}
	\mathcal L^\dagger=\mathcal L.
\end{equation}
Hermiticity is a stronger condition than normality and therefore automatically satisfies $\eta(\mathcal L) = [\mathcal L,\mathcal L^\dagger]=0$.
Consequently, all results established for globally normal Liouvillians remain valid in this limit.

Hermiticity imposes a stronger structural constraint than global normality,
since the coherent and dissipative contributions can no longer be chosen
independently. Therefore using \erf{eqn:GKSL}, the Hermiticity condition, \erf{eqn:Hermitian:L}, is equivalent to
\begin{equation}
	\sum_k \left(\hat L_k\rho \hat L_k^\dagger - \hat L_k^\dagger\rho\hat L_k \right) = i 2[\hat H,\rho],
\end{equation}
or alternatively
\begin{equation}
	(\mathrm{Im}\,\mathcal J)(\rho) = [\hat H,\rho],
\end{equation}
by introducing the anti-Hermitian part of the jump superoperator $\mathrm{Im}\,\mathcal J \equiv ({\mathcal J-\mathcal J^\dagger})/{2i}$. Thus, the Hamiltonian commutator is exactly reproduced by the anti-Hermitian component of the jump superoperator. In particular, if $\mathcal J$ is Hermitian then $[\hat H,\rho]=0$,
for every density operator $\rho$, implying that the Hamiltonian is
proportional to the identity. Consequently, the coherent and dissipative
contributions cannot be specified independently, but must satisfy an exact
operator identity.

An important subclass is therefore obtained when
$\hat H=0$ (or more generally $\hat H\propto\hat I$). In this case the
Hermiticity condition reduces simply to $\mathcal J^\dagger=\mathcal J$,
so that the dissipative contribution is itself Hermitian in Liouville space.
This includes, for example, pure dephasing models generated by Hermitian jump
operators, for which each $\hat L_k=\hat L_k^\dagger$.

The spectral structure simplifies considerably. Since $\mathcal L$ is Hermitian, all eigenvalues are real $\Lambda_\alpha \in \mathbb{R}$, and the Liouvillian admits a complete orthonormal eigenoperator basis satisfying
\begin{equation}
	\mathcal L(\rho_\alpha)
	=
	-\Gamma_\alpha \rho_\alpha,
\end{equation}
with $\Gamma_\alpha \in \mathbb{R}_{\ge 0}$. Expanding the unconditional density operator as in \erf{eqn:rho:expand:rho_a} the modal amplitudes evolve independently according to \erf{eqn:ME:coef:sol} with $\Lambda_\alpha = \Gamma_\alpha$.
Unlike a generic normal Liouvillian, which may support oscillatory modes through complex eigenvalues, a Hermitian Liouvillian generates purely relaxational dynamics. Each orthogonal mode decays independently without Liouville-space rotations or phase accumulation. Consequently, every timescale of the unconditional dynamics is determined solely by the relaxation spectrum $\{\Gamma_\alpha\}$, without effects associated with eigenoperator non-orthogonality or coherent Liouville-space rotations spectral phases.

This provides the simplest realization of the modal-compressibility arguments developed previously. Since no transfer of population can occur between orthogonal eigenoperators, the active information content of the unconditional state decreases monotonically as rapidly decaying modes are suppressed. The long-time dynamics is therefore governed entirely by the slowest-decaying terms of the spectrum. 

At the trajectory level, however, the situation remains more subtle. The stochastic coefficients associated with an individual quantum trajectory continue to satisfy the jump-coupled stochastic evolution equations derived in~\srf{subsec:single_trajectory_dynamics}. Consequently, individual trajectories may still exhibit mixing between Liouvillian eigenmodes through the stochastic jump process. Hermiticity therefore removes modal coupling from the ensemble-averaged dynamics but does not eliminate stochastic mode mixing at the level of a single realization. Since all Liouvillian eigenvalues are purely real, the deterministic contribution to each modal amplitude exhibits only exponential relaxation. The oscillatory phase evolution characteristic of generic normal Liouvillians is therefore absent, although stochastic jump events continue to couple different Liouvillian modes within an individual realization. As before, the independent modal evolution is recovered only after ensemble averaging, \erf{eqn:ME:coef:avg}.

The second-moment analysis also simplifies substantially. In the doubled Liouville space, the uncoupled generator possesses eigenvalues
\begin{equation}
	\Lambda_{\alpha\beta}
	=
	-(\Gamma_\alpha+\Gamma_\beta),
\end{equation}
which are purely real and nonpositive. The corresponding propagator is therefore a strict contraction,
\begin{equation}
	\left\|
	e^{t\widetilde{\mathcal L}^{(2)}_0}
	\right\|
	\le 1,
\end{equation}
precluding any transient amplification within the uncoupled second-moment dynamics. Hermitian Lindbladians therefore constitute the most restrictive subclass of globally normal generators. Both the unconditional dynamics and the uncoupled doubled-space evolution reduce to purely dissipative contractions along mutually orthogonal Liouville-space directions.

\subsubsection{Spin-1 pure dephasing}

As an illustrative example, consider a spin-1 system undergoing pure
dephasing,
\begin{equation}
	\hat H =0,
	\qquad
	\hat L=\sqrt{\gamma_{\rm d}}\,\hat S_z ,
\end{equation}
where $\gamma_{\rm d} \in \mathbb{R}^+$ is the dephasing rate and
\begin{equation}
	\hat S_z=
	\begin{pmatrix}
		1&0&0\\
		0&0&0\\
		0&0&-1
	\end{pmatrix}.
\end{equation}
Since $\hat L=\hat L^\dagger$, the jump superoperator $\mathcal J(\rho) = \gamma_{\rm d}\,\hat S_z\rho\hat S_z$ is Hermitian. Furthermore, because the Hamiltonian vanishes, the smooth generator is likewise Hermitian, and consequently $\mathcal D_{\mathcal J}=\mathcal D_{\mathcal S}=\mathcal D_{\mathcal SJ}=0$. Therefore, the Hermiticity condition, \erf{eqn:Hermitian:L}, is satisfied identically. The resulting master equation,
\begin{equation}
	\dot\rho = \gamma_{\rm d} \left(\hat S_z\rho\hat S_z - \frac12 \left\{ \hat S_z^2,\rho \right\} \right),
\end{equation}
contains no coherent evolution and therefore generates purely irreversible decay of the off-diagonal density-matrix elements. The Liouvillian spectrum is therefore entirely real and nonpositive\footnote{In the operator basis $\{\ket{m}\bra{n}\}$ where $m,n \in (1,0,-1)$ and $\ket{m}$ are eigenstate of $\hat S_z$, the Lindblad is expressed as $\mathcal L (\ket{m}\bra{n}) = \Gamma_{mn} \ket{m}\bra{n}$. Hence every operator is an eigenoperator of the Liouvillian with eigenvalue $\Gamma_{mn}=-\gamma_{\rm d}(m-n)^2/2$. The spectrum therefore consists of three stationary population modes, four coherences decaying at rate $\gamma_{\rm d}/2$, and two coherences decaying at rate $2\gamma_{\rm d}$.}, in agreement with the general discussion above. From the modal perspective, each orthogonal Liouvillian eigenoperator decays independently without any oscillatory evolution or transient
mixing. 

At the trajectory level, the stochastic unraveling remains nontrivial. Between jumps the conditioned state evolves under $\hat H_{\rm eff}=-i\gamma_{\rm d}\hat S_z^2/2$. The resulting non-Hermitian attenuation is removed by trajectory normalization, leaving no coherent rotation. The jump operation applies the Hermitian operator $\hat S_z$ followed by
normalization. Therefore, individual trajectories remain free of coherent
mode rotations, although the nonlinear normalization of stochastic updates
can still couple Liouvillian coefficients at the trajectory level.
The ensemble average removes these realization-dependent couplings and
recovers the independent exponential relaxation of the Hermitian
Liouvillian modes.

\subsection{Structured Dissipative Processes}

A second class of analytically tractable systems arises when the jump operators satisfy the uniform-loss condition
\begin{equation} \label{eqn:cond:strctured:L}
	\sum_k \hat L_k^\dagger \hat L_k
	= \Omega \hat I,
\end{equation}
where $\Omega \in \mathbb{R}^+$ is a constant. This structure appears whenever the total dissipation rate is independent of the instantaneous system state. Under this condition, the effective non-Hermitian Hamiltonian reduces to
\begin{equation}
	\hat H_{\rm eff} = \hat H-\frac{i}{2}\Omega \hat I,
\end{equation}
and consequently $[\hat H_{\rm eff},\hat H_{\rm eff}^\dagger]=0$.
Therefore, the smooth deterministic evolution generated by $\mathcal S$ is itself normal $\mathcal D_{\mathcal S}=0$. The normality condition therefore simplifies to
\begin{equation} \label{eqn:eta:structured:L}
	\eta(\mathcal{L}) = \|\mathcal D_{\mathcal J} + \mathcal D_{\mathcal SJ}\|.
\end{equation}
In contrast to the generic situation, where normality emerges through a three-way balance between smooth evolution, jump processes, and their interference, the present class of systems requires only a cancellation between the jump events components and those of the smooth-jump interference. Consequently, all sources of non-normality originate from the jump operators themselves.

The effective Hamiltonian also acquires a particularly simple propagator,
\begin{equation} \label{eqn:smooth:propag:struct:L}
	e^{-i \hat H_{\rm eff}t} = e^{-\Omega t/2} e^{-i \hat Ht},
\end{equation}
showing that the deterministic evolution consists of a unitary rotation accompanied by a uniform exponential attenuation. Consequently, \begin{equation}
	\left\| e^{-i \hat H_{\rm eff}t} \right\| = e^{-\Omega t/2}.
\end{equation}
To contrast this result with a general bound, we consider the following lemma regarding the spectral norm of a non-normal matrix exponential:
\begin{lemma} \label{lemma:smooth:propagator}
	Let $\hat H_{\rm{eff}} = \hat H - \frac{i}{2} \sum_k \hat L_k\dg \hat L_k$ be a diagonalizable, non-normal matrix with right-eigenvector matrix $V$ and diagonal eigenvalue matrix $D = \text{diag}(\lambda_1, \dots, \lambda_D)$, where $\lambda_m = \varepsilon_m - \frac{i}{2}\gamma_m$. The spectral norm of its propagator satisfies:
	\begin{equation}
		\label{eq:lemma_statement}
		\left\| e^{-i H_{\rm{eff}} \Delta t} \right\| \le \Bbbk_V e^{-\gamma_{\min}\Delta t/2},
	\end{equation}
	where $\Bbbk_V \equiv \parallel V \parallel \parallel V^{-1} \parallel$ is the eigenvector condition number, and $\gamma_{\min}$ is the minimum decay eigenvalue.
\end{lemma}

The details of proof can be found in \arf{appn:proof:smooth:lemma}.
 In generic non-normal systems, $\Bbbk_V>1$ allows transient amplification even when all eigenmodes are individually decaying. Here, since $\hat H_{\rm eff}$ is normal, its eigenvectors are orthogonal, implying $\Bbbk_V=1$. Furthermore all decay rates coincide, $\gamma_m=\Omega$, so that $\gamma_{\min}=\Omega$. The general bound of Lemma is therefore saturated exactly. Hence, the deterministic propagator is a strict contraction for all times, and cannot generate the non-orthogonal mode mixing associated with pseudospectral amplification.

The condition in \erf{eqn:cond:strctured:L} admits an especially transparent interpretation at the trajectory level. Since every pure state experiences the same total loss rate,
\begin{equation}
	\sum_k \wp_k(t)=\Omega,
\end{equation}
the total jump rate becomes completely state independent. Waiting times are therefore distributed according to a Poisson process with fixed rate $\Omega$, eliminating state-dependent variations in the jump statistics.

Moreover, the conditioned stochastic master equation undergoes a further simplification yielding
\begin{equation}
	d\rho_c = -i[\hat H,\rho_c]\,dt+ \sum_k dN_k\left(\frac{\mathcal{J}_k(\rho_c)}{\wp_k}-\rho_c\right).
\end{equation}
Consequently, the deterministic evolution of every trajectory reduces to purely unitary motion generated by $H$, while all stochasticity is confined to the jump events. In the trajectory-coefficient equations, this cancellation appears through the exact compensation between the dissipative drift and the normalization term\footnote{The deterministic term coefficient must be precisely the projection of the unitary commutator, that is $-i \,{\rm Tr} (\rho_\alpha^\dagger([\hat H, \rho_c]))$.}. Thus, for structured dissipative processes satisfying \erf{eqn:cond:strctured:L}, the trajectory dynamics separates into deterministic unitary transport punctuated by stochastic jumps occurring at a constant state-independent rate.

This cancellation provides a particularly transparent example of the interference mechanisms discussed throughout this work. Although the
underlying open-system dynamics remains dissipative, the non-Hermitian drift generated by the loss operators is exactly cancelled by the normalization contribution of the conditioned evolution. As a result, dissipation manifests exclusively through stochastic jump events, while the deterministic trajectory segments evolve unitarily.

An even simpler limit is obtained when the coherent Hamiltonian vanishes. The effective Hamiltonian then becomes proportional to the identity, $	\hat H_{\rm eff} = -\frac{i}{2}\Omega \hat I$, and the smooth evolution reduces to a uniform contraction. In this purely dissipative regime, all nontrivial dynamics originates from the jump maps. The normality condition becomes entirely a statement about the jump subspace, $\eta(\mathcal L)=\|\mathcal D_{\mathcal J}\|$, since both the smooth contribution and the smooth-jump interference term vanish identically.

\subsubsection{Driven thermal qubit}
Consider a two-level system with Hamiltonian $\hat H_0 = \frac{\nu}{2}\hat \sigma_z$, where $\nu$ denotes the transition frequency between the two energy levels and $\hat \sigma_z$ is the Pauli-$z$ operator. The system is coupled to an infinite-temperature thermal environment through jump operators 
\begin{equation}
	\hat L_\pm=\sqrt{\gamma_{\rm th}}\,\hat\sigma_\pm,
\end{equation}
corresponding respectively to excitation and relaxation processes with transition rate $\gamma_{\rm th}$. Here $\hat\sigma_\pm = (\hat \sigma_x\pm i\hat \sigma_y)/2$ are raising $(+)$ and lowering $(-)$ operators, and $\hat \sigma_x$ and $\hat \sigma_y$ are Pauli-$x$ and Pauli-$y$ operators. The total dissipative strength satisfies \erf{eqn:cond:strctured:L} with $\Omega = \gamma_{\rm th}$, implying the smooth contribution is always normal $\mathcal D_{\mathcal S}=0$. In the Pauli basis $\{\hat I,\hat\sigma_x,\hat\sigma_y,\hat\sigma_z\}$, the Liouvillian assumes the block-diagonal form
\begin{equation}
	\mathcal L=
	\begin{pmatrix}
		0&0&0&0\\
		0&-\gamma_{\rm th}&-\nu&0\\
		0&\nu&-\gamma_{\rm th}&0\\
		0&0&0&-2\gamma_{\rm th}
	\end{pmatrix},
\end{equation}
which is readily verified to satisfy $[\mathcal L,\mathcal L^\dagger]=0$. The coherent dynamics acts only within the $(\hat\sigma_x,\hat\sigma_y)$ subspace through an orthogonal rotation, while each invariant subspace decays independently.
For this model the jump superoperator is also Hermitian, so that $[\mathcal J,\mathcal J^\dagger]=0$. The example therefore represents a particularly simple member of the structured-dissipator class in which both the smooth and jump generators are individually normal. Accordingly, the mixed smooth--jump contribution also vanishes, and the Liouvillian normality condition is satisfied term by term rather than through the cancellation mechanism discussed in the general theory.

At the trajectory level, the effective Hamiltonian generates the propagator in \erf{eqn:smooth:propag:struct:L}, whose norm is exactly $e^{-\gamma_{\rm th} t/2}$. Since $\hat H_{\rm eff}$ is normal, the eigenvectors remain orthogonal and the condition number appearing in Lemma~\ref{lemma:smooth:propagator} is unity. After normalization, however, the common exponential attenuation cancels identically from the stochastic master equation, leaving only the unitary evolution generated by $\hat H_0$. Individual trajectories therefore consist of coherent precession about the $z$ axis interrupted by stochastic excitation and relaxation events occurring with the constant Poisson rate $\gamma_{\rm th}$. The stochastic jumps redistribute population between the energy eigenstates but do not generate the trajectory-level non-normal alignment discussed in Sec.~\ref{subsec:wavefunction_geometry}. This model therefore illustrates the structured-dissipation limit in its simplest form, where the absence of non-orthogonal mode mixing induced by stochastic map reflects the fact that each dynamical phase of the unraveling is already normal individually.

\section{Conclusion} \label{sec:conclusion}

In this work we investigated how Liouvillian normality manifests itself at the
level of quantum trajectories. While previous studies classified Markovian open
quantum systems according to the non-normality of the full Lindblad generator,
we showed that this global property admits a natural trajectory-level
interpretation through the interplay between deterministic smooth evolution and
stochastic quantum jumps. By decomposing the Lindbladian into smooth and jump
generators and analyzing the associated commutator structure, we found that
Liouvillian normality is generally not inherited by the individual trajectory
generators. Instead, it is maintained through an exact algebraic balance between
the smooth, jump, and mixed smooth--jump contributions. In the steady-state
subspace this balance reduces to an exact equality between the expectation
values of the smooth and jump commutators, identifying the mixed
smooth--jump term as the mechanism through which global Liouvillian normality is
enforced. At the spectral level, global normality also excludes defective
Liouvillian spectra, ruling out exceptional points and guaranteeing an
orthogonal decomposition of the relaxation dynamics.

Expanding both the conditional and unconditional dynamics in the Liouvillian
eigenoperator basis further clarified the distinction between individual
trajectories and ensemble evolution. For normal Lindbladians, the unconditional
state evolves through independent orthogonal relaxation modes, whereas
individual stochastic trajectories generally exhibit mode mixing induced by the
jump process. These trajectory-level couplings disappear exactly after ensemble
averaging, recovering the independent modal evolution of the unconditional
density operator.

The second-moment analysis further showed that, for normal Lindbladians, transient amplification associated with the uncoupled doubled dynamics is absent. Within this doubled-space description, any growth of trajectory fluctuations originates from the stochastic jump coupling rather than from non-orthogonal Liouvillian eigenmodes. Consequently, the long-time dynamics is governed entirely by the orthogonal relaxation modes of the normal Liouvillian.

Finally, we examined two analytically tractable subclasses of normal
Lindbladians. Hermitian Lindbladians represent the most restrictive case,
supporting purely relaxational dynamics with real spectra, while structured
dissipative processes reduce the deterministic evolution to unitary transport
between stochastic jumps occurring at a constant Poisson rate. These examples
illustrate how the general framework simplifies in physically relevant limits
while preserving the underlying trajectory-level interpretation.

Several directions remain open. Since the decomposition into smooth and jump
contributions depends on the chosen unraveling, it would be interesting to
determine whether unravelings can be optimized to reduce trajectory-level mode
mixing or improve simulation efficiency without altering the unconditional
dynamics. A second challenge is the construction of genuinely generic normal
Lindbladians for which the smooth, jump, and mixed commutators are all
nonvanishing while satisfying the normality condition. Such models
would provide valuable testbeds for the mechanisms identified here. Finally, it
would be interesting to explore whether trajectory-level non-normality can be
exploited as a resource for quantum simulation, feedback, or error-mitigation
protocols, where the complexity of individual trajectories, rather than the
unconditional dynamics, is often the primary computational bottleneck.

	\section*{Acknowledgment}
This work was partly supported by the Australian Government through the Trailblazer Program.

	\appendix
	\numberwithin{equation}{section}

\section{Alternative proof for Theorem~\ref{theorem:noEP}} \label{appn:proof:noEP}

	Assume for contradiction that $\mathcal{L}$ possesses an exceptional point at a specific choice of system parameters. Under this assumption, $\mathcal{L}$ becomes defective, meaning there exists a degenerate eigenvalue $\Lambda_0$ associated with a nontrivial Jordan block of at least order 2. Let ${\rho}_0$ be the true, nonzero eigenoperator at the base of the Jordan chain, and let ${\rho}_1$ be its generalized eigenoperator. By definition of a Jordan block, these operators satisfy the coupled relations:
	\begin{align}
		\label{eq:jordan_base}
		\mathcal{L}(\rho_0) &= \Lambda_0 \rho_0, \\
		\label{eq:jordan_gen}
		\mathcal{L}(\rho_1) &= \Lambda_0 \rho_1 + \rho_0.
	\end{align}
	We evaluate the structural constraint of Eq.~\eqref{eq:jordan_gen} by taking its Hilbert-Schmidt inner product with the base eigenoperator $\rho_0$:
	\begin{equation} \label{eq:ep_inner_prod}
		\langle \rho_0, \mathcal{L}(\rho_1) \rangle = \Lambda_0 \langle \rho_0 \vert \rho_1 \rangle + \langle \rho_0 \vert \rho_0 \rangle.
	\end{equation}
	Utilizing the definition of the superoperator adjoint, the left-hand side of Eq.~\eqref{eq:ep_inner_prod} can be rewritten by shifting the action of $\mathcal{L}$ onto the left state element:
	\begin{equation} \label{eq:ep_adjoint_shift}
		\langle \rho_0, \mathcal{L}(\rho_1) \rangle = \langle \mathcal{L}^\dagger(\rho_0), \rho_1 \rangle.
	\end{equation}
	Crucially, because $\mathcal{L}$ is a normal superoperator ($[\mathcal{L}, \mathcal{L}^\dagger] = 0$), the spectral theorem for normal operators dictates that $\mathcal{L}$ and $\mathcal{L}^\dagger$ share the exact same eigenoperators with complex conjugate eigenvalues. Substituting $\mathcal{L}^\dagger(\rho_0) = \Lambda_0^* \rho_0$ into Eq.~\eqref{eq:ep_adjoint_shift} yields:
	\begin{equation}
		\langle \mathcal{L}^\dagger(\rho_0), \rho_1 \rangle = \langle \Lambda_0^* \rho_0 , \rho_1 \rangle = \Lambda_0 \langle \rho_0 \vert \rho_1 \rangle,
	\end{equation}
	where the complex conjugate scalar $\Lambda_0^*$ pulls out of the left-hand side of the inner product as a standard linear multiplier $\Lambda_0$.
	Equating this result back into the right-hand side expansion of Eq.~\eqref{eq:ep_inner_prod} results in:
	\begin{equation}
		\Lambda_0 \langle \rho_0 \vert \rho_1 \rangle = \Lambda_0 \langle \rho_0 \vert \rho_1 \rangle + \parallel \rho_0 \parallel^2.
	\end{equation}
	Subtracting the scalar term $\Lambda_0 \langle \rho_0 \vert \rho_1 \rangle$ from both sides produces a direct mathematical contradiction:
	\begin{equation}
		\parallel \rho_0 \parallel^2 = 0.
	\end{equation}
	Because $\rho_0$ is a nonzero eigenoperator of the system, its Hilbert-Schmidt norm must be strictly positive ($\parallel \rho_0 \parallel^2 > 0$). The assumption that a Jordan block can exist under global normality yields $0 > 0$, which is impossible. Consequently, $\mathcal{L}$ is fundamentally forbidden from possessing a defective Jordan structure, proving that normal Lindbladians are strictly immune to exceptional points.

\section{Derivation of the Doubled Liouvillian for Unnormalized Jump Trajectories}
\label{app:doubled_liouvillian}

In order to analyze the stability of trajectory fluctuations, we introduce an
unnormalized jump unraveling for which the stochastic evolution remains linear
in the trajectory density operator. The unnormalized state
$\tilde{\rho}_c$ obeys the stochastic master equation
\begin{equation} 	\label{eq:unnormalized_jump_SME}
	d\tilde{\rho}_c = \left( \mathcal{S}(\tilde{\rho}_c) + K\tilde{\rho}_c \right)dt + \sum_{k=1}^{K} \left(\mathcal{J}_k(\tilde{\rho}_c) - \tilde{\rho}_c \right)dN_k,
\end{equation}
where $\mathcal{J}_k$ denotes the jump superoperator associated with channel
$k$, and $\mathcal{S}$ contains the deterministic no-jump contribution. The
Poisson increments satisfy the It\^o rules
\begin{equation} \label{eqn:ito:rule}
	\mathbb{E}[dN_k]=dt ,
\end{equation}
and
\begin{equation} \label{eqn:jump:orthogonal}
	dN_k dN_m
	=
	\delta_{km}dN_k .
\end{equation}
Terms of order $dt^2$, $dt\,dN_k$, and higher orders are neglected in the
continuous-time limit. Identifying the physical Lindblad generator as

\begin{equation} 	\label{eq:Lindblad_decomposition}
	\mathcal{L}
	=\mathcal{S}+\sum_k\mathcal{J}_k \equiv \mathcal{S}+\mathcal{J},
\end{equation}
the first moment reproduces the standard master equation.

To study the second moment, we introduce the doubled trajectory operator

\begin{equation} 	\label{eq:unnormalized_doubled_state}
	\widetilde{\bm{\varrho}}(t) = \tilde{\mathbb{E}} \left[ \tilde{\rho}_c(t)\otimes\tilde{\rho}_c(t) \right].
\end{equation}
Using the It\^o product rule, one obtains: 
\begin{align}
	d(\tilde{\rho}_c\otimes\tilde{\rho}_c) &= \left[(\mathcal{S}+K\mathcal{I})\otimes\mathcal{I} + \mathcal{I}\otimes(\mathcal{S}+K\mathcal{I})\right] dt + \nn \\
	&\sum_k \left[\mathcal{J}_k - \mathcal{I}) \otimes \mathcal{I} + \mathcal{I}\otimes(\mathcal{J}_k -\mathcal{I})\right] dN_k + \nn \\
	& \sum_k
	(\mathcal{J}_k-\mathcal{I})
	\otimes
	(\mathcal{J}_k-\mathcal{I})
	(\tilde{\rho}_c\otimes\tilde{\rho}_c) dN_k.
\end{align}
The last term originates from the quadratic variation of the Poisson process, where using \erf{eqn:jump:orthogonal} only same jump channel products survive. 
After averaging over the Poisson process, the doubled state therefore evolves
according to a closed linear equation,
\begin{equation} 	\label{eq:doubled_linear_evolution}
	\frac{d}{dt}\widetilde{\bm{\varrho}}(t) = \widetilde{\mathcal L}^{(2)} \widetilde{\bm{\varrho}}(t),
\end{equation}
with 
\begin{equation}
	\widetilde{\mathcal L}^{(2)} =\mathcal{S}\otimes\mathcal{I}+\mathcal{I}\otimes\mathcal{S}+\sum_{k=1}^K\mathcal{J}_k\otimes\mathcal{J}_k + K\,\mathcal{I}\otimes\mathcal{I}
\end{equation}
Using the Lindblad decomposition, \erf{eq:Lindblad_decomposition}, the doubled generator may be written
in the form

\begin{equation}
	\widetilde{\mathcal L}^{(2)} = \mathcal L\otimes\mathcal I + \mathcal I\otimes\mathcal L + \widetilde{\mathcal W},
\end{equation}
where 
\begin{equation}
	\widetilde{\mathcal W} = \sum_{k=1}^K (\mathcal J_k - \mathcal{I)} \otimes (\mathcal J_k - \mathcal{I)},
\end{equation}
 contains all stochastic correlation terms generated by the jump process.

This doubled representation is the starting point for analyzing whether
second-moment amplification and therefore trajectory variance growth can occur
under global Liouvillian normality.

\section{Trajectory-weight fluctuations} \label{appn:wieght:fluc}

The linear doubled-space formulation admits a direct interpretation in terms of stochastic trajectory weights. Defining the unnormalized trajectory weight
\begin{equation}
	w_c(t) = \mathrm{Tr} \left[ \tilde{\rho}_c(t) \right],
\end{equation}
we first derive its stochastic evolution directly from the unnormalized stochastic master equation,
\begin{equation}
	d\tilde{\rho}_c = \left( \mathcal S(\tilde{\rho}_c) + K\tilde{\rho}_c \right)dt + \sum_{k=1}^{K} \left( \mathcal J_k(\tilde{\rho}_c) - \tilde{\rho}_c \right)dN_k.
\end{equation}
Taking the trace of both sides yields
\begin{equation}
	dw_c = \left( \mathrm{Tr} [\mathcal S(\tilde{\rho}_c)] + Kw_c \right)dt + \sum_{k=1}^{K} \left( \mathrm{Tr} [\mathcal J_k(\tilde{\rho}_c)] - w_c \right)dN_k .
\end{equation}
Using the trace-preserving property of the Lindblad generator,
\begin{equation}
	\mathrm{Tr} [\mathcal S(\tilde{\rho}_c)] = - \mathrm{Tr} [\mathcal J(\tilde{\rho}_c) ], \qquad \mathcal J = \sum_k\mathcal J_k,
\end{equation}
together with the definition
\begin{equation}
	r_k={\mathrm{Tr}[\mathcal J_k(\tilde{\rho}_c)]}/{w_c},
\end{equation}
the weight evolution becomes
\begin{equation} \label{eq:weight_SDE}
	dw_c = w_c(K-\sum_k r_k)\,dt + w_c \sum_k (r_k-1)\,dN_k 
\end{equation}

To characterize the growth of weight fluctuations, we apply the It\^{o} product rule,
\begin{equation}
	d(w_c^2)=2w_c\,dw_c+(dw_c)^2 .
\end{equation}
Using Eq.~\eqref{eq:weight_SDE} together with \erfa{eqn:ito:rule}{eqn:jump:orthogonal}
we obtain
\begin{align}
	d(w_c^2) & = 2w_c^2 (K-\sum_k r_k )\,dt + 2w_c^2 \sum_k (r_k-1) \,dN_k \nonumber\\
	&\quad + w_c^2 \sum_k (r_k-1)^2 \,dN_k .
\end{align}
Taking the ensemble average and using $\mathbb E[dN_k]=dt$, one obtains
\begin{equation} 	\label{eq:weight_second_moment}
	\frac{d}{dt}\tilde{\mathbb E}[w_c^2]=\tilde{\mathbb E}\left[w_c^2 \sum_k (r_k-1)^2\right].
\end{equation}
Equation~\eqref{eq:weight_second_moment} governs the growth of the second moment of the stochastic trajectory weights.
For finite-dimensional systems, the instantaneous jump rates are bounded because the normalized trajectory state satisfies $\rho_c\ge0$ and $\mathrm{Tr}(\rho_c)=1$. Consequently,
\begin{equation}
	r_k = \mathrm{Tr} \left( L_k^\dagger L_k \rho_c \right) \le \left\|L_k^\dagger L_k\right\|,
\end{equation}
 It follows that
\begin{equation}
	(r_k-1)^2
	\le
	\max
	\left\{
		1,
		\left(\left\|L_k^\dagger L_k\right\|-1\right)^2
		\right\},
\end{equation}
and therefore
\begin{equation}
	\sum_{k=1}^{K}(r_k-1)^2 \le C_L,
\end{equation}
with
\begin{equation}
	C_L = \sum_{k=1}^{K} \max \left\{ 1, \left(\left\|L_k^\dagger L_k\right\|-1\right)^2 \right\}.
\end{equation}
Using Eq.~(\ref{eq:weight_second_moment}), the second moment of the trajectory weights satisfies
\begin{equation}
	\frac{d}{dt}\mathbb{E}[w_c^2] = \mathbb{E} \left[w_c^2 \sum_k (r_k-1)^2 \right] \le C_L\,\mathbb{E}[w_c^2].
\end{equation}
An application of Gr\"{o}nwall's inequality yields
\begin{equation}
	\mathbb{E}[w_c^2(t)] \le	e^{C_L t}\, \mathbb{E}[w_c^2(0)].
\end{equation}
Thus, although the correlation superoperator $\widetilde{\mathcal W}$ can generate growth of trajectory-weight fluctuations, its contribution is bounded by at most an exponential envelope determined solely by the jump operators. In particular, $\widetilde{\mathcal W}$ cannot produce super-exponential growth of the trajectory weights, excluding an uncontrolled amplification mechanism within the doubled-space dynamics.

The same result follows directly from the doubled-space evolution equation. Using
\begin{equation}
	\tilde{\bm{\varrho}}(t) = \mathbb E \left[\tilde{\rho}_c(t) \otimes \tilde{\rho}_c(t) \right],
\end{equation}
together with the identity
\begin{equation}
	\mathrm{Tr}(X\otimes Y) = \mathrm{Tr}(X)\,\mathrm{Tr}(Y),
\end{equation}
one finds
\begin{equation} 	\label{eq:R_trace_weight}
	\mathrm{Tr} \left[\tilde{\bm{\varrho}}(t) \right] = \tilde{\mathbb E} [w_c(t)^2].
\end{equation}
Taking the trace of the doubled evolution equation,
\begin{equation}
	\dot{\tilde{\bm{\varrho}}} = \widetilde{\mathcal L}^{(2)} \tilde{\bm{\varrho}},
\end{equation}
gives
\begin{equation}
	\frac{d}{dt} \tilde{\mathbb E}[w_c^2] = \mathrm{Tr} \left[\widetilde{\mathcal L}^{(2)} \tilde{\bm{\varrho}} \right].
\end{equation}
Substituting
\begin{equation}
	\widetilde{\mathcal L}^{(2)} = \widetilde{\mathcal L}^{(2)}_0 + \widetilde{\mathcal W},
\end{equation}
reveals that the uncoupled contribution does not affect weight growth. Indeed,
\begin{equation}
	\mathrm{Tr} \left[(\mathcal L\otimes\mathcal I) \tilde{\bm{\varrho}} \right]=0, \qquad \mathrm{Tr} \left[(\mathcal I\otimes\mathcal L) \tilde{\bm{\varrho}} \right]=0,
\end{equation}
because the Lindblad generator is trace preserving,
\begin{equation}
	\mathrm{Tr} [\mathcal L(X)] = 0 \qquad \forall \, X.
\end{equation}
Consequently,
\begin{equation}
	\mathrm{Tr} \left[\widetilde{\mathcal L}^{(2)}_0 \tilde{\bm{\varrho}} \right]=0,
\end{equation}
and therefore
\begin{equation} 	\label{eq:weight_growth_from_W}
	\frac{d}{dt} \tilde{\mathbb E}[w_c^2] = \mathrm{Tr} \left[\widetilde{\mathcal W} \tilde{\bm{\varrho}} \right].
\end{equation}
Equation~\eqref{eq:weight_growth_from_W} shows that all growth of the trajectory-weight fluctuations originates from the correlation superoperator $\widetilde{\mathcal W}$. In contrast, the uncoupled doubled Liouvillian $\widetilde{\mathcal L}^{(2)}_0$ contributes only to the geometric evolution of the doubled state and cannot generate weight amplification. This separation forms the basis of the stability analysis presented in the main text.

\section{Proof of Lemma~\ref{lemma:smooth:propagator}} \label{appn:proof:smooth:lemma}

	By applying a similarity transformation, the effective Hamiltonian is expressed in its eigenbasis as $\hat{H}_{\text{eff}} = V D V^{-1}$ where $D = \text{diag}(\lambda_1, \dots, \lambda_D)$, and $\lambda_m = \varepsilon_m - \frac{i}{2}\gamma_m$. Taking the matrix exponential of both sides yields:
	\begin{equation}
		e^{-i \hat{H}_{\text{eff}} \,\Delta t} = e^{-i (V D V^{-1})\, \Delta t} = V e^{-i \Delta t D}\, V^{-1}.
	\end{equation}
	Taking the induced matrix norm of both sides and invoking the sub-multiplicative property of matrix norms results in the inequality:
	\begin{equation}
			\left\| e^{-i \hat{H}_{\text{eff}} \Delta t} \right\| = \left\| V e^{-i D \Delta t} V^{-1} \right\| \le \left\| V \right\| \cdot \left\| e^{-i \Delta t D} \right\| \cdot \left\| V^{-1} \right\|.
	\end{equation}

	Regrouping the boundary terms isolates the eigenvector condition number $\Bbbk_V$:
	\begin{equation}
		\left\| e^{-i \hat{H}_{\text{eff}} \Delta t} \right\| \le \Bbbk_V \left\| e^{-i \Delta t D} \right\|.
	\end{equation}
	Because $D$ is strictly diagonal, the matrix norm of its exponential is exactly equal to the maximum absolute value of its diagonal entries: 
	\begin{equation}
		\left\| e^{-i D \Delta t} \right\| = \max_m \big| e^{-i(\varepsilon_m - \frac{i}{2}\gamma_m)\Delta t} \big|.
	\end{equation}
	Since $|e^{-i \varepsilon_m \Delta t}| = 1$ for all real energy components, the magnitude is governed purely by the real decay exponent, which is maximized by selecting the minimum decay rate in the spectrum:
	\begin{equation}
		\left\| e^{-i D \Delta t} \right\| = \max_m \left( e^{-\gamma_m \Delta t / 2} \right) = e^{-\gamma_{\min} \Delta t / 2}.
	\end{equation}
	Substituting this back into the norm inequality yields the finished bound:
	\begin{equation}
		\left\| e^{-i \hat{H}_{\text{eff}} \Delta t} \right\| \le \Bbbk_V\, e^{-\gamma_{\min} \Delta t / 2}.
	\end{equation}

	\bibliography{references}
	
\end{document}